\documentclass{amsart}

\usepackage{amssymb,amsmath,euscript,hyperref}
\hypersetup{colorlinks,breaklinks,
             urlcolor=blue,
             linkcolor=blue}

\usepackage{}

\newtheorem{theorem}{Theorem}[section]
\newtheorem{lemma}[theorem]{Lemma}
\newtheorem{proposition}[theorem]{Proposition}

\newtheorem{corollary}[theorem]{Corollary}

\theoremstyle{definition}
\newtheorem{definition}[theorem]{Definition}
\newtheorem{example}[theorem]{Example}

\theoremstyle{remark}
\newtheorem{remark}[theorem]{Remark}
\def\P{\mathcal{P}}
\def\Fq{\mathbf{F}_q}

\def\Fq{{\mathbb F}_q}
\def\Fqn{{\mathbb F}_q^n}

\def\AA{{\mathbb A}}
\def\FF{{\mathbb F}}
\def\PP{{\mathbb P}}
\newcommand{\Ev}{\operatorname{Ev}}
\newcommand{\codim}{\operatorname{codim}}
\newcommand{\wt}{\operatorname{wt}}

\newcommand{\supp}{\mathrm{Supp}}

\def\PP{{\mathbb P}}

\def\Z{\mathbb{Z}}

\begin{document}
	
\title{Pure Resolutions, Linear Codes, and Betti Numbers}
	
	
	\author{Sudhir R. Ghorpade}
	\address{Department of Mathematics, 
		Indian Institute of Technology Bombay,\newline \indent
		Powai, Mumbai 400076, India}
\email{\href{mailto:srg@math.iitb.ac.in}{srg@math.iitb.ac.in}}
\thanks{Sudhir Ghorpade is partially supported by DST-RCN grant INT/NOR/RCN/ICT/P-03/2018 from the Dept. of Science \& Technology, Govt. of India, MATRICS grant MTR/2018/000369 from the Science \& Engg. Research Board, and IRCC award grant 12IRAWD009 from IIT Bombay.}
	\author{Prasant Singh}
\address{Department of Mathematics and Statistics, 
		UiT The Arctic University of Norway, 
\newline \indent
N-9037 Troms{\o}, Norway}
\thanks{Prasant Singh is supported by a postdoctoral  fellowship  from  grant 280731 from the Research Council of Norway.}
\email{\href{mailto:psinghprasant@gmail.com}{psinghprasant@gmail.com}}
	

	\date{\today}
	
	\begin{abstract}
		We consider the minimal free resolutions of Stanley-Reisner rings associated to linear codes and give an intrinsic characterization of linear codes having a pure resolution.
We use this characterization to quickly deduce 
the minimal free resolutions of Stanley-Reisner rings associated to  MDS codes as well as constant weight 
codes. 
We also deduce that the minimal free resolutions of Stanley-Reisner rings of 
first order Reed-Muller codes
are pure, and explicitly describe the Betti numbers. Further, we  show that in the case of higher order Reed-Muller codes, the minimal free resolutions are almost always not pure.  
The nature of the minimal free resolution of Stanley-Reisner rings corresponding to several classes of two-weight codes, besides the first order Reed-Muller codes, is also determined. 
%
	\end{abstract}
	
	\maketitle
	\section{Introduction}
One of the interesting 
 developments  in algebraic coding theory in the recent past is the association of a fine set of invariants, called \emph{Betti numbers}, to linear error correcting codes. This is due to Johnsen and Verdure \cite{JV1} and their idea is as follows. Some basic terminology used below is reviewed in the next section. 

Let $C$ be a $q$-ary linear code of length $n$ and dimension $k$ and let $H$ be a parity check matrix of $C$. 
The vector matroid corresponding to $H$ is a pure simplicial complex, say $\Delta$, and its Stanley-Reisner ring $R_\Delta$ over $\Fq$ is 
a finitely generated standard graded  $\Fq$-algebra of dimension $n-k$. As such it has a minimal graded free resolution. Moreover,  $\Delta$ is shellable, thanks to a classical result that goes back to Provan \cite{Provan} (see also Bj\" orner 
\cite[\S 7.3]{AB1}). Hence $R_\Delta$ is Cohen-Macaulay. (See, for example, \cite[Ch. 6, \S 2]{GSSV}).  
So by the Auslander-Buchsbaum formula, the length of any minimal free resolution of $R_\Delta$ is $n - (n-k) = k$, and it looks like 
\begin{equation}
\label{eq:resol}
F_k {\longrightarrow}F_{k-1}\longrightarrow \cdots \longrightarrow F_1\longrightarrow F_0 \longrightarrow R_{\Delta} \longrightarrow 0
\end{equation}
where $F_0= R:= \Fq[X_1, \dots , X_n]$ and each $F_i$ is a graded free $R$-module of the form 
\begin{equation}
\label{eq:FiBetaij}
F_i = \bigoplus_{j \in \mathbb{Z}}R(-j)^{\beta_{i,j}} \quad \text{for } i=0, 1, \dots , k.
\end{equation}
The nonnegative integers $\beta_{i,j}$ thus obtained depend only on $C$ (and not on the choice of $H$ or the minimal free resolution  of $R_\Delta$), and are the \emph{Betti numbers} of~$C$. Thus we may refer to \eqref{eq:resol} as a (graded minimal free) \emph{resolution} of $C$. Such a resolution is said to be \emph{pure} of type $(d_0, d_1, \ldots, d_k)$ if for each $i=0,1, \dots , k$, the Betti number $\beta_{i,j}$ is nonzero if and only if $j=d_i$. If, in addition, 
$d_0, d_1, \dots , d_k$ are consecutive, 
then the resolution is said to be \emph{linear}.   Johnsen and Verdure \cite{JV1} showed that the Betti numbers of a code $C$ contain information about all the generalized Hamming weights $d_i(C)$ of $C$. In fact, they showed that  
\begin{equation}
\label{eq: BettiGHW}
d_i(C) =\min\{j:\beta_{i,j}\neq 0\} \quad \text{for } i=1, \dots , k.
\end{equation}
More recent work of Johnsen, Roksvold and Verdure \cite{JRV} shows that the Betti numbers of $C$ and its elongations determine the so-called generalized weight polynomial of $C$. Thus, if we 
combine this with the results of Jurrius and Pellikaan \cite{JP}, then we obtain a direct relation between  the generalized weight enumerator of $C$ and the Betti numbers  of $C$ and of its elongations. 

It is clear therefore that explicit determination of Betti numbers of codes would be useful and interesting. On the other hand, it is usually a hard problem, except in some special cases. 
The simplest class of codes for which Betti numbers are completely determined 
is that of MDS codes where the minimal free resolution is linear. The next case is that of simplex codes or dual Hamming codes, which are essentially the prototype of constant weight codes (indeed, by a classical result of Bonisoli \cite{Bo}, every constant weight code is a concatenation of simplex codes, possibly with added $0$-coordinates). 
For such codes, the Betti numbers were explicitly determined by Johnsen and Verdure in another paper \cite{JV2}. In this case, it turns out that the resolution is pure, although not necessarily linear. 

In general, Betti numbers of pure resolutions are relatively easy to determine, thanks to a formula of Herzog and K{\"{u}}hl  \cite{HK}, which in the case of linear codes provides an expression for the Betti numbers in terms of the generalized Hamming weights. So the result for simplex codes can be deduced from it if one knows that their (minimal free) resolutions are necessarily pure. 
Partly with this in view, we consider the question of obtaining an intrinsic characterization for a linear code to have a pure resolution. A complete characterization is given in Theorem~\ref{thm: pure}. 
This is then applied to show that the first order Reed-Muller codes have a pure resolution and all their Betti numbers can be described explicitly. On the other hand, we show that 
Reed-Muller codes of order $2$ or more do not, in general, 
have a pure resolution. As a corollary, it is seen that the property of admitting a pure resolution is not preserved when passing to the dual. 

The first order Reed-Muller codes are examples of two-weight codes, and it is natural to ask if a similar result holds for every two-weight code. However, unlike constant weight codes, the structure of two-weight codes is far more complicated and it is a topic of considerable research in coding theory and finite projective geometry. We refer to the survey of Calderbank and Kantor \cite{CK} and the references therein for a variety of examples of two-weight codes. We also take up the question of determining the Betti numbers of many of these codes. It is seen that the resolution is not always pure and thus we can not appeal to the Herzog-K{\"{u}}hl formula. Nonetheless, we succeed in determining the Betti numbers of many two-weight codes, partly by  using a set of equations due to Boij and S\" oderberg \cite{BS}. It appears that the technique of Boij-S\" oderberg equations used here could be fruitful in the determination of Betti numbers of many important classes of linear codes. 

We remark that although our results on the Betti numbers of simplex and first order Reed-Muller codes using the Herzog-K{\"{u}}hl formula were obtained independently in early 2015, Trygve Johnsen \cite{J} has informed us that similar formulas are obtained in the Ph.D. thesis of Armenoff \cite{Arm} and the Master's thesis of Karpova \cite{AK}. 
In any case, our emphasis here is on the general characterization of purity, applications to Reed-Muller codes that are not only of first order, but also of higher order, and the determination of Betti numbers of many two-weight codes.

\section{Preliminaries}
Fix, throughout this paper, positive integers $n,k$ with $k\le n$ and a finite field $\Fq$ with $q$ elements. We denote by 
$[n]$ the set $\{1, \ldots, n\}$ of first $n$ positive integers. Also, $2^{[n]}$ denotes the set of all subsets of $[n]$. For any finite set $\sigma$, we denote by $|\sigma|$ the cardinality of $\sigma$. 
By a $[n, k]_q$-code, we shall mean a $q$-ary linear code of length $n$ and dimension $k$, i.e., a $k$-dimensional subspace of $\Fq^n$. 

\subsection{Codes and Matroids}\label{subsec:2.1}
Let $C$ be a $[n, k]_q$-code and
let $H$ be a parity check matrix of $C$.
For $i\in [n]$, let $H_i$ denote the $i$-th column of $H$. 
Define 
\[
\Delta :=
\{\sigma \in 2^{[n]} : \{H_i : i\in \sigma\} \text{ is linearly  independent over }\Fq\}.
\]
The ordered pair $({[n]},\Delta)$ is a matroid, and we call it the \emph{matroid associated to the code} $C$. Elements of $\Delta$ are called \emph{independent sets} of this matroid. A maximal independent  set in $\Delta$ is called a \emph{basis} of the matroid. 
It is well-known that every basis of a matroid has the same cardinality and this number is called the \emph{rank} of the matroid. 
 If $\sigma\subseteq {[n]}$ 
and if we let 
$\Delta|\sigma := \{ \tau\in \Delta: \tau \subseteq \sigma\}$, then  $(\sigma,\Delta|\sigma)$ is a matroid, called the \emph{restriction} of the matroid $({[n]},\Delta)$ to $\sigma$; the rank of this restricted matroid is called the \emph{rank} of 
$\sigma$ and denoted by $r(\sigma)$; 
the difference $|\sigma|-r(\sigma)$ is denoted by $\eta(\sigma)$ and called the \emph{nullity} of $\sigma$. Evidently, the rank of the matroid $({[n]},\Delta)$ is the rank of $H$, which is $n-k$, and so the nullity of any  $\sigma\subseteq {[n]}$ ranges from $0$ to $k$. 
For $0\le i \le k$, 
we define 
\[
N_i:=\{\sigma\subseteq {[n]}: \eta(\sigma)=i\}. 
\]

\subsection{Stanley-Reisner Rings and Betti Numbers}
Suppose $({[n]},\Delta)$ is as in the previous subsection. Then $\Delta$ is a simplicial complex. We denote by $I_\Delta$ the ideal of the polynomial ring $R:= \Fq[X_1,\ldots, X_n]$ generated by all monomials of the form $\prod_{i\in \tau} X_i$, where $\tau\in 2^{[n]}\setminus \Delta$. The quotient $R_\Delta=R/I_\Delta$  is called the \emph{Stanley-Reisner ring} or the \emph{face ring} associated to $\Delta$. As noted in the Introduction, $R_\Delta$ has a minimal free resolution of the form \eqref{eq:resol}. Furthermore, since $I_\Delta$ is a monomial ideal generated by squarefree monomials, we can choose the free $R$-modules $F_i$ in  \eqref{eq:resol} to be not only $\Z$-graded as in \eqref{eq:FiBetaij}, but also $\Z^n$-graded so as to write
\begin{equation}
\label{eq:FiBetaSigma}
F_i = \bigoplus_{\sigma \in \mathbb{Z}^n}R(-\sigma)^{\beta_{i,\sigma}}  \quad \text{for } i=1, \dots , k.
\end{equation}
In fact, the $\Z^n$-graded Betti numbers $\beta_{i,\sigma}$ have the property that $\beta_{i,\sigma} = 0$ unless the $n$-tuple $\sigma = (\sigma_1, \dots , \sigma_n)$ has all its coordinates in $\{0,1\}$. Such $n$-tuples in $\{0,1\}^n$ can be naturally identified with subsets of ${[n]}$ where $(\sigma_1, \dots , \sigma_n)$ corresponds to the subset $\{i\in {[n]}: \sigma_i=1\}$ of ${[n]}$ that we shall also denote by $\sigma$. 
Thus, 
we may index the direct sum in  \eqref{eq:FiBetaSigma} by $\sigma \in 2^{[n]}$.  The relation between the $\Z$-graded and $\Z^n$-graded Betti numbers is simply that 
%
\begin{equation}
\label{eq: betaij-betasigma}
 \beta_{i,j}=\sum_{|\sigma|=j}\beta_{i,\sigma} \quad \text{for } i=1, \dots , k.
\end{equation}
Johnsen and Verdure \cite{JV1} proved an important relationship between the $\Z^n$-graded Betti numbers and subsets of a given nullity. Namely, for $1\le i \le k$ and $\sigma \subseteq {[n]}$, 
\begin{equation}
\label{eq: nonzerobetti}
\beta_{i,\sigma}\neq 0\iff \sigma\in N_i\text{ and }\sigma \text{ is a minimal element of }N_i.
\end{equation}
This result, which can perhaps be traced back to Stanley \cite[p. 59]{St},  will be very useful for us in the sequel. 
Note also that if $\mu_1, \dots , \mu_t$ are squarefree monomials in $R$ which constitute a minimal set of generators of $I_{\Delta}$ and if $\sigma_j \in 2^{[n]}$ denotes the support of $\mu_j$ (so that $\mu_j = \prod_{i \in \sigma_j} X_i$) for $1\le j \le t$, then without loss of generality, we can take the first free $R$-module in \eqref{eq:FiBetaSigma} to be 
\begin{equation}
\label{eq:F1}
F_1 = \bigoplus_{j=1}^t R(-\sigma_j).
\end{equation}

Finally, we recall the following general result, which was alluded to in the Introduction. A proof can be found in \cite{BS}. We note that a graded module $M$ over a polynomial ring $R$ having projective dimension $k$ will have a minimal free resolution such as \eqref{eq:resol} with $R_\Delta$ replaced by $M$, except in this case $F_0$ may not be equal to $R$.  In general, we let $\beta_i:= \mathrm{rank}_R(F_i) = \sum_{j} \beta_{i,j}$. Note that  if $M$ has a pure resolution of type $(d_0, d_1, \dots , d_k)$, then $\beta_i := \beta_{i, d_i}$ for $i=0,1, \dots , k$. 
\begin{theorem}[Boij-S\" oderberg]\label{boijsoderberg}
	Let $R$ be the 
polynomial ring over a field and let $M$ be a graded $R$-module of finite projective dimension $k$. Then $M$ is  Cohen-Macaulay if and only if its graded Betti numbers satisfy the equations
	\begin{equation}
	\label{eq: boijsoderberg}
	\sum\limits_{i=0}^k\sum\limits_{j\in\mathbb{Z}}(-1)^ij^\ell\beta_{i,j}=0\quad \text{for }\ell=0,\ldots, k-1.
	\end{equation}
	In particular, if the minimal free resolution of $M$ is pure of type $(d_0, d_1,\ldots, d_k)$,
then 
$\eqref{eq: boijsoderberg}$ implies the Herzog-K{\"{u}}hl formula \cite{HK}:
	\begin{equation}
	\label{eq: herzogkuhl}
	\beta_i =\beta_0\left\vert\prod\limits_{j \neq i} \frac{d_j}{(d_j-{d_i})}\right\vert \quad \text{for } i= 1,\ldots,k, 
	\end{equation}
\end{theorem}

As noted in the Introduction, Stanley-Reisner rings associated to  linear codes (or more generally, simplicial complexes corresponding to matroids) are Cohen-Macaulay, and hence the above theorem is applicable; moreover, in this 
case, $\beta_0=1$. If a $[n,k]_q$-code $C$ has a pure resolution of type $(d_0, \dots , d_k)$, then $d_0=0$ and for $1\le i \le k$,  $d_i$ is precisely the $i$-th generalized Hamming weight of $C$, thanks to \eqref{eq: BettiGHW}; we will refer to $\beta_i  = \beta_{i, d_i}$ as the \emph{Betti numbers of $C$} in this case. 

\section{Pure Resolution of Linear Codes}
In this section we will give a characterization of the purity of the resolution of the Stanley-Reisner ring associated to a linear code in terms of the support weight of certain subcodes of the code. We will then outline some simple applications. 

Let $C$ be a $[n,k]_q$-code and let $H=[H_1 \ldots H_n]$ be a parity check matrix of $C$, where, as before, $H_i$ denotes the $i${th} column of $H$.
For any  subset $\sigma $ of $[n]$, define  
$S(\sigma)$ to be the subspace $\langle H_i: i\in \sigma \rangle$ of $\mathbb{F}_q^{n-k}$ spanned by the columns of $H$ indexed by $\sigma$. Note that $r(\sigma) = \dim S(\sigma)$. 
Let us also define 
a related subspace  of $\Fq^n$ by 
$$
\widehat{S}(\sigma):=\{x=(x_1, \dots , x_n)\in \mathbb{F}_q^n: x_i=0 \mbox{ for }i\notin\sigma\mbox{ and }\sum_{i\in \sigma} x_iH_i=0\}.
$$ 
Recall that for any subcode $D$ of $C$, i.e., a subspace $D$ of $C$, 
 the \emph{support} of $D$ is the set $\supp(D)$ of all $i\in [n]$ for which there is 
$x=(x_1, \dots , x_n) \in D$ with $x_i\ne0$;
further, 
we let $\wt(D):=|\supp(D)|$, and call this the \emph{weight} of $D$. 

\begin{lemma}\label{support}
	Let $\sigma\subseteq [n]$. Then $\widehat{S}(\sigma)$ is a subcode of $C$ and $\supp(\widehat{S}(\sigma))\subseteq \sigma$.
\end{lemma}

\begin{proof}
Since 
$C=\{x=(x_1, \dots , x_n) \in \Fq^n : \sum_{i=1}^n x_iH_i =0\}$, it is clear that 
$\widehat{S}(\sigma)$ is a subcode of $C$. The inclusion $\supp(\widehat{S}(\sigma))\subseteq \sigma$ is obvious.
\end{proof}
For any $\sigma\subseteq [n]$, let $ \mathbb{F}_q^{\sigma}$ denote the set of all ordered $|\sigma|$-tuples $(x_i)_{i\in \sigma}$ of elements of $\Fq$ indexed by $\sigma$. Consider the map
\begin{equation}
\label{eq: linmap}
\phi_{\sigma}:  \mathbb{F}_q^{\sigma}\to S(\sigma) \quad \text{defined by} \quad \phi_{\sigma}(x)=\sum_{i\in \sigma} x_iH_i.
\end{equation}
Clearly $\phi_{\sigma}$ is a surjective $\Fq$-linear map. 

\begin{lemma}
	Let $\sigma\subseteq [n]$ and let $\phi_{\sigma}$ be as in \eqref{eq: linmap}. Then $\ker\phi_{\sigma}$ is isomorphic (as a $\Fq$-vector space) to $\widehat{S}(\sigma)$. Consequently, 
\begin{equation}
\label{eq: 505}
\dim S(\sigma)= \lvert\sigma\rvert-\dim  \widehat{S}(\sigma).
\end{equation}
\end{lemma}
\begin{proof}
Consider the map $\psi:  \mathbb{F}_q^{\sigma}\longrightarrow \mathbb{F}_q^n$ given by $\psi(x)=(v_1,v_2,\ldots,v_n)$, where
	$$ 
v_i = \begin{cases} 
	x_i &  \mbox{if }i\in\sigma, \\
	0 & \mbox {otherwise.}
	\end{cases} 
	$$
	It is easily seen that the restriction of $\psi$ to $\ker \phi_\sigma$  gives an isomorphism of $\ker \phi_\sigma$ onto $\widehat{S}(\sigma)$. The second assertion follows from the Rank-Nullity theorem. 
\end{proof}
For $0\le i\le k$, let 
$\mathbb{G}_i(C)$ denote the Grassmannian of all $i$-dimensional subspaces of $C$. We call $D\in \mathbb{G}_i(C)$ an  \emph{$i$-minimal subcode} of $C$ if $\supp(D)$ is minimal among the supports of all $i$-dimensional subcodes of $C$, i.e., $\supp(D^\prime)\nsubseteq\supp(D)$ for any $D^\prime\in \mathbb{G}_i(C)$ with $D^\prime\neq D$. We let 
\[
\mathcal{D}_i= \text{the set of all $i$-minimal subcodes of }C.
\]
Note that if $i=0$, then the only element of $\mathbb{G}_i(C)$ is $\{0\}$, and its support is $\emptyset$, which is clearly $i$-minimal. Moreover, $r(\emptyset) = 0=|\emptyset|$, and thus $\supp(\{0\}) \in N_0$.  
 In fact, a more general result holds. 
Recall (from $\S\;\ref{subsec:2.1}$) that $N_i$ denotes the set of all subsets of $[n]$ of nullity $i$.
\begin{proposition}\label{pr}
Suppose $0\le i \le k$ and $D\in\mathcal{D}_i$. Then $\supp(D)\in N_i$.
\end{proposition}

\begin{proof}
	Let $\sigma:=\supp(D)$. Then for any $x\in D$,  clearly $x_i=0$  for all $i\in [n]$ with $i\notin \sigma$.  Also, since $D\subseteq C$, we see that $\sum x_iH_i=0$ for each $ x=(x_1, \dots , x_n) \in D$. It follows that 
$D\subseteq \widehat{S}(\sigma)$. In particular, 
$\dim \widehat{S}(\sigma)\geq i$. Further, by Lemma \ref{support},
$$
\sigma=\text{supp}(D)\subseteq \text{supp}(\widehat{S}(\sigma))\subseteq\sigma.
$$	
Therefore $\text{supp}(\widehat{S}(\sigma))=\sigma$.
In case $\dim (\widehat{S}(\sigma))> i$, we can choose some $j\in \sigma$ and observe that $\{x\in \widehat{S}(\sigma) : x_j=0\}$ is a subspace of dimension $\dim \widehat{S}(\sigma) -1$, and its support is contained in $\sigma\setminus\{j\}$. This can be used to 
construct an $i$-dimensional subcode $D'$ of $\widehat{S}(\sigma)$ with support a proper subset of $\sigma$. But then 
the minimality of the support of $D$ is contradicted.  It follows that $\dim \widehat{S}(\sigma)= i$, and hence $D= \widehat{S}(\sigma)$. 
Now equation \eqref{eq: 505} shows that $r(\sigma)=\lvert\sigma\rvert-i$, that is, $\sigma\in N_i$.
%
%
\end{proof}

It turns out that a partial converse of the above proposition is also true.

\begin{proposition}\label{pr2}
Suppose $0\leq i\leq k$ and  $\sigma$ is a minimal element of $N_i$ (with respect to inclusion).  Then there exists $D\in\mathcal{D}_i$
such that $\sigma=\supp(D)$. 
	
\end{proposition}
\begin{proof}
Since $\sigma\in N_i$, we see that $\dim S(\sigma)=r(\sigma)=\lvert\sigma\rvert - i$.
Hence, equation \eqref{eq: 505} implies that 
$\dim \widehat{S}(\sigma)=i$. Let $D:=\widehat{S}(\sigma)$ and $\sigma' :=\supp(D)$. Then $D$ is an $i$-dimensional subcode of $C$ and by Lemma \ref{support}, 
	$\sigma' \subseteq \sigma$. We claim that $D\in\mathcal{D}_i$. To see this, assume the contrary. Then 
there exists $D^\prime\in \mathbb{G}_i(C)$ with $D^\prime\neq D$ such that $\supp(D^\prime)\subsetneq\sigma'$. 
Replacing $D'$ by an $i$-dimensional subcode with smaller support, if necessary, we may assume that $D'$ is $i$-minimal. But then by Proposition \ref{pr}, $\supp(D^\prime)\in N_i$, which 
contradicts the minimality of 
$\sigma$ in $N_i$.  Thus, $D\in\mathcal{D}_i$.
\end{proof}

\begin{corollary}
\label{cor:3.5}
Suppose $0\leq i\leq k$ and $\sigma\subseteq [n]$. Then $\sigma$ is a minimal element of $N_i$ 
if and only if there exists an $i$-minimal subcode $D$ of 	$C$ with $\supp(D)=\sigma$.
\end{corollary}
\begin{proof}
Follows from Propositions \ref{pr} and \ref{pr2}.
\end{proof}

\begin{theorem}\label{thm: pure}
	Let $C$ be an $[n,k]_q$ code and $\,  d_1 < \dots < d_k$  its 
generalized Hamming weights. 
Then any 
$\mathbb{N}$-graded minimal free resolution of $C$ is 
pure if and only if for each $i=1, \dots , k$, all the $i$-minimal subcodes of $C$ have support weight 
	$d_i$. 
\end{theorem}
\begin{proof} 
From \eqref{eq: nonzerobetti} and Corollary~\ref{cor:3.5}, we see that for $1\leq i\leq k$ and $\sigma\subseteq [n]$,
\begin{equation}\label{eq:betaisigma}
\beta_{i,\sigma}\neq 0  \iff  \sigma = \supp(D) \text{ for some } D\in \mathcal{D}_i. 
\end{equation}
Thus, the desired result follows from \eqref{eq: BettiGHW} 
and \eqref{eq: betaij-betasigma}.
\end{proof}

\begin{corollary}
	\label{Cor:beta1}
The Betti numbers at the first step of a $[n,k]_q$-code $C$  are given by 
$$
\beta_{1,j} = \left| \{ D\in \mathcal{D}_1 : \wt(D) = j\}\right| \ \text{for any nonnegative integer } j.
$$
\end{corollary}
\begin{proof} 
Follows from \eqref{eq:F1} and \eqref{eq:betaisigma}. 
\end{proof}

\begin{remark}
\label{rem:hsteps}
{\rm 
Let $C$ be an $[n,k]_q$ code and $h$ a positive integer $\le k$. Given a resolution of $C$, say \eqref{eq:resol}, by its \emph{left part  after $h$ steps}, we mean the exact sequence
$$
F_k\longrightarrow F_{k-1}\longrightarrow\cdots \longrightarrow F_h 
$$
which is a minimal free resolution of the cokernel of the last map $F_{h+1} \longrightarrow F_h $.  Now let $\,  d_1 < \dots < d_k$   be the 
generalized Hamming weights of $C$. It is clear that the proof of Theorem~\ref{thm: pure} also shows that the 
 left part after $h$ steps of any $\mathbb{N}$-graded minimal free resolution of $C$ is pure 
if and only if for each $i=h, \dots , k$, all the $i$-minimal subcodes of $C$ have support weight $d_i$.
%
}
\end{remark}

We now show how a characterization due to Johnsen and Verdure \cite{JV1} of MDS codes can be deduced from our characterization of purity, and moreover,  how  the minimal free resolution of an MDS code can then be readily determined using the Herzog-K{\"{u}}hl formula. 

\begin{corollary}
	\label{Cor: MDS}
Let $C$ be a nondegenerate $[n,k]_q$-code and $h$ a positive integer $\le k$.
Then $C$ is $h$-MDS if and only if the left part of its resolution after $h$ steps is linear. 
In particular, $C$ is an MDS code if and only if its resolution is linear. Moreover, if $C$ is MDS, then its Betti numbers are given by 
$$
\beta_i=\binom{n-k+i-1}{i-1}\binom{n}{k-i} \quad \text{for } i=1, \dots , k.
$$
\end{corollary}
\begin{proof}
	Suppose the left part after $h$ steps of a resolution of $C$ is linear. Since $C$ is  nondegenerate, $d_k=n$, and so from the linearity 
together with equation \eqref{eq: BettiGHW},  we obtain 
$d_i= n-k+i$ for $h\le i \le k$. 
Taking $i=h$, we see that 
$C$ is $h$-MDS.  

Conversely, suppose $C$ is $h$-MDS. Then from the strict monotonicity of generalized Hamming weights \cite[Thm. 1]{We}, we see that $d_i = n-k+i$ for $h\le i \le k$. Now fix $i\in \{h, \dots , k\}$ and let $D$ be an $i$-minimal subcode of $C$. Let $\sigma:= \supp(D)$. By Proposition~\ref{pr}, $\sigma \in N_i$. Also, $n-k+i = d_i \le  |\sigma|$. Consequently, $n-k \le |\sigma| - i = r(\sigma) \le n-k$. It follows that $|\supp(D)| =d_i$. Thus, in view of Remark \ref{rem:hsteps}, we conclude that the left part after $h$ steps of any resolution of $C$ is linear.

%
%
Now assume that $C$ is MDS. Then, in view of 
\eqref{eq: herzogkuhl}, we see that for $1\le i \le k$, 
$$
\beta_i =\prod\limits_{j\neq i}\frac{d_j}{\lvert d_j-d_i\rvert}  =\prod\limits_{j\neq i}\frac{n-k+j}{\lvert j-i\rvert} = \bigg(\prod_{j=1}^{i-1} \frac{n-k+j}{i-j} \bigg) 
 \bigg(\prod_{j=i+1}^{k} \frac{n-k+j}{j-i} \bigg), 
$$
and an easy calculation shows that this  is 
equal to  $\binom{n-k+i-1}{i-1}\binom{n}{k-i}$.
%
%
\end{proof}

Let us also show how the result of Johnsen and Verdure \cite{JV2} about the minimal free resolution of constant weight codes can be deduced  from Theorem~\ref{thm: pure}. 
\begin{corollary}
	\label{cor: constwt}
	Let $C$ be an $[n,k]_q$-code in which each nonzero codeword has constant weight $d$. 
	Then the $\mathbb{N}$-graded resolution of $C$ is pure. Moreover,  the generalized Hamming weights (or the shifts) and the Betti numbers of $C$ are given by 
	$$
d_i=\frac{q^{k-1}(q^i-1)}{q^{i-1}(q-1)}  \quad \text{and} \quad 
\beta_i= {k\brack i}_q \displaystyle{q^{\frac{i(i-1)}{2}}},  \quad \text{for } i=1, \dots , k, 
$$
where ${k\brack i}_q$ denotes 
the Gaussian binomial coefficient.
\end{corollary}

\begin{proof}
It is well-known (see, e.g., \cite[Thm. 1]{LC}) that every $j$-dimensional subcode of the constant weight code $C$ has support weight $d_j$, where 
$$
	d_j =\frac{d(q^j-1)}{q^{j-1}(q-1)}  \quad \text{for } j=1, \dots , k. 
$$
Hence, by Theorem $\ref{thm: pure}$,  $C$ has a pure resolution. Evidently, the numbers $d_i$ defined above are the generalized Hamming weights of $C$. Moreover, for $i,j=1, \dots , k$, 
$$
d_i - d_j = \frac{d(q^{i-j}-1)}{q^{i-1}(q-1)},  \quad \text{if $j<i$, whereas} \quad  d_j - d_i = \frac{d(q^{j-i}-1)}{q^{j-1}(q-1)},  \quad \text{if $j> i$}. 
$$
Hence, the Herzog-K{\"{u}}hl formula $\eqref{eq: herzogkuhl}$ implies that for $i=1, \dots , k$, 
$$
\beta_i =\prod\limits_{j\neq i}\frac{d_j}{\lvert d_j-d_i\rvert}  = \bigg(\prod_{j=1}^{i-1}\frac{q^{i-j} (q^j-1)}{q^{i-j}-1}\bigg) 
 \bigg(\prod_{j=i+1}^{k} \frac{q^{j}-1}{q^{j-i}-1} \bigg) =  \displaystyle{q^{\frac{i(i-1)}{2}}} {k\brack i}_q,
$$
where the last equality follows by noting that for $i=1, \dots , k$, 
$$
{k\brack i}_q = {k\brack{k- i}}_q = \frac{(q^k-1) (q^{k-1}-1) \cdots (q^{i+1}-1)}{(q^{k-i}-1) (q^{k-i-1}-1) \cdots (q-1)}
= \prod_{j=i+1}^k \frac{q^{j}-1}{q^{j-i}-1}.
$$
This proves the desired result. 
\end{proof}

\section{Reed-Muller Codes}
In this section we consider generalized Reed-Muller codes and prove that the resolution of the first order Reed-Muller code is pure, whereas for other Reed-Muller codes, it is non-pure. 
Let us begin by recalling the construction of  (generalized) Reed-Muller codes. 
Fix 
integers $r,m$ such that  $m\ge 1$ and $0\le r \leq m (q-1)$. 
Define 
$$
V_q(r, m)=\{f\in \Fq[X_1,\ldots, X_m] : \deg f\leq r \text{ and }  \deg_{X_i} f< q\text{ for }i=1, \dots , m\}.
$$
Fix an ordering $P_1, \ldots, P_{q^m}$ of the elements of $\Fq^m$. Consider the evaluation map
$$
\Ev: V_q(r, m)\to \Fq^{q^m} \quad \text{defined by} \quad f \mapsto c_f:=\left(f(P_1),\ldots, f(P_{q^m})\right).
$$
The image of $\Ev$ 
is called the \emph{generalized Reed-Muller code of order $r$},  and we 
denote it by $\mathcal{RM}_q(r, m)$. It is well-known that 
$\mathcal{RM}_q(r, m)$ is an 
$[n, k, d]_q$-code, with
\begin{equation}\label{eq:RMnkd}
n= q^m,\quad k=\sum_{i=0}^m (-1)^i \binom{m}{i} \binom{m+r-iq}{m}, \quad \text{and}\quad d= (q-s)q^{m-t-1},
\end{equation}
where $t,s$ are unique integers satisfying $r= t(q-1) +s$ and $0\leq s\leq q-2$. 
Further, for any $ \omega_0, \omega_1, \dots , \omega_t\in \Fq$ with $\omega_0\ne 0$ and any distinct 
$\omega_1', \dots, \omega_s'\in \Fq$, the polynomial 
\[
f(X_1\ldots, X_m)= \omega_0\prod_{i=1}^{t}(1-(X_i-\omega_i)^{q-1})\prod_{j=1}^{s}(X_{t+1}-\omega^\prime_j)
\]
is in $V_q(r, m)$ and $\Ev(f)$ is a minimum weight codeword of $\mathcal{RM}_q(r, m)$. Moreover, up to a (nonhomogeneous) linear substitution in $X_1, \dots , X_m$, every minimum weight codeword of $\mathcal{RM}_q(r, m)$ is of this form; see, e.g., Theorems 2.6.2 and 2.6.3 of \cite{DGM}. 
It is also well-known (see, e.g., \cite[\S 5.4]{AssKey}) that the dual of $\mathcal{RM}_q(r, m)$ is given by\footnote{Strictly speaking, for the formula \eqref{eq:RMdual} to be valid, we should  note that the definition of $\mathcal{RM}_q(r, m)$ is meaningful also when $r=-1$ in which case it is the zero code of length $q^m$.}
\begin{equation}\label{eq:RMdual}
\mathcal{RM}_q(r, m)^\perp = \mathcal{RM}_q(r^\perp, m) \quad \text{where} \quad r^\perp + r +1 = m(q-1).
\end{equation}
In particular, if $r=m(q-1)-1$, then $\mathcal{RM}_q(r, m)$ is a MDS code (being the dual of $\mathcal{RM}_q(0, m)$, which is the $1$-dimensional code of length $q^m$  generated by the all-$1$ vector). Also if $r=m(q-1)$, then $\mathcal{RM}_q(r, m)$ is a MDS code, being the full space $\Fq^m$. Finally, if $m=1$, then $\mathcal{RM}_q(r, m)$ is a Reed-Solomon code, and in particular, a MDS code. Thus, in all these ``trivial cases", $\mathcal{RM}_q(r, m)$  has a pure, and in fact, linear, resolution. The following result deals with the first nontrivial case of $r=1$. 
%
%

\begin{theorem}
	\label{thm: firstReedMuller}
	The $\mathbb{N}$-graded minimal free resolution of the first order Reed-Muller code  $\mathcal{RM}_q(1,m)$  is pure and is given by
$$ 
R(-d_{m+1})^{\beta_{m+1}} {\longrightarrow}R(-d_m)^{\beta_m}\longrightarrow \cdots \longrightarrow R(-d_1)^{\beta_1} \longrightarrow R
$$
where $d_i = q^m-\lfloor q^{m-i}\rfloor$ for $1\le i\le m+1$, and 
$$ 
\beta_i=\begin{cases} 
\displaystyle
q^{\binom{i+1}{2}} 
 \prod_{j=1}^{m-i}\frac{q^{m+1-j} -1}{q^{m+1-i-j}-1} & \text{ if }1\le  i\leq m, \\
\displaystyle \prod_{j=1}^{m}(q^j-1) & \mbox { if }i=m+1. \end{cases} 
$$ 
\end{theorem}

\begin{proof}
First, note that $\dim \mathcal{RM}_q(1,m) = m+1$. Let $i$ be a positive integer $\le m+1$. If $i=m+1$, then the only $i$-dimensional subcode of $\mathcal{RM}_q(1,m)$ is $\mathcal{RM}_q(1,m)$ itself, and this has support weight $q^m$. Now suppose $1\leq i\leq m$. 
	Let $D$ be a subcode of  $\mathcal{RM}_q(1,m)$ of dimension $i$. Then the support weight of $D$ is clearly 
$$
q^m-|Z(f_1,\ldots,f_i)|, 
$$
where $f_1, \dots, f_i \in V_q(1, m)$ are linearly independent polynomials whose images under $\Ev$ form a basis of $D$, and where $Z(f_1,\ldots,f_i)$ denotes the set of common zeros in $\Fq^m$ of $f_1, \dots , f_i$. Now $f_1 = \cdots = f_i=0$ is a system of $i$ linearly independent (not necessarily homogeneous) linear equations in $m$ variables, and thus it has either no solutions (when the system is inconsistent) or exactly $q^{m-i}$ solutions (when the system is consistent). Accordingly, the support weight of $D$ is either $q^m$ or $q^m - q^{m-i}$. Moreover, if the former holds, then $\supp(D) = \{1, \dots , q^m\}$, and so $D$ cannot be an $i$-minimal subcode of $\mathcal{RM}_q(1,m)$. It follows that all $i$-minimal subcodes of $\mathcal{RM}_q(1,m)$ have the same support weight $d_i = q^m-\lfloor q^{m-i}\rfloor$ for $1\le i\le m+1$. Thus, by Theorem $\ref{thm: pure}$, $\mathcal{RM}_q(1,m)$ has a pure resolution. Consequently, the 
Betti numbers of $\mathcal{RM}_q(1,m)$ can be 
determined using the 
Herzog-K{\"{u}}hl formula~\eqref{eq: herzogkuhl} as follows. 
$$
\beta_{m+1} =\prod\limits_{j=1}^m\frac{d_j}{d_{m+1}-d_j} = 
		\prod\limits_{j=1}^m\frac{q^m - q^{m-j}}{q^{m-j}} = \prod_{j=1}^{m}(q^j-1),
$$
whereas for $1\le i \le m$, 
\begin{eqnarray*}
		\beta_i &=&\frac{d_{m+1}}{d_{m+1}-d_i}\prod\limits_{m+1>j>i}\frac{d_j}{d_j-d_i}\prod\limits_{j<i}\frac{d_j}{d_j-d_i}\\
		&=&q^i\prod\limits_{j={i+1}}^m\frac{q^{m-j}(q^{j}-1)}{q^{m-j}(q^{j-i}-1)}\prod\limits_{j=1}^{i-1}\frac{q^{m-j}(q^j-1)}{q^{m-i}(q^{i-j}-1)}\\
		&=&  q^{\frac{i(i+1)}{2}} \prod\limits_{j={i+1}}^m \frac{(q^{j}-1)}{(q^{j-i}-1)} 
= q^{\binom{i+1}{2}}  \prod\limits_{j=1}^{m-i}\frac{(q^{m+1-j}-1)}{(q^{m+1-i-j}-1)}.
\end{eqnarray*}
This proves the theorem. 
\end{proof}

\begin{remark}
Observe that 
the pure resolution of $\mathcal{RM}_q(1,m)$  in Theorem \ref{thm: firstReedMuller} is linear only when either $m=1$ or $m=2=q$. As noted earlier, $\mathcal{RM}_q(1,m)$ is a MDS code in this case. 
\end{remark}

Next, we shall show that the minimal free $\mathbb{N}$-resolutions of many generalized Reed-Muller codes of order higher than one are not pure. It will be convenient to consider various cases separately. As usual, we shall say that an element $c$ of a linear code $C$ is a \emph{minimal codeword} if either $c=0$, or if $c\ne 0$ and the support of the $1$-dimensional subspace $\langle c \rangle$ of $C$ spanned by $c$ is minimal among the supports of all  $1$-dimensional subcodes of $C$. Evidently, a codeword of minimum weight is minimal, but the converse may not be true. 

\subsection{Binary Case}\label{subsec:binary} 
In this subsection we consider the binary case, i.e., when $q=2$. 
We will use the following simple, but useful, observation. It is stated, for instance, in Ashikhmin and Barg \cite[Lemma 2.1]{AB2}. The proof is obvious and is omitted. 

\begin{lemma}
\label{lem:AB2} 
Let $C$ be a binary linear code and let $d= d(C)$ be its minimum distance. If $c\in C$ is not a minimal weight codeword, then $c = c_1+c_2$ for some nonzero $c_1, c_2\in C$ such that 
$\supp (\langle c_1 \rangle)$ and $\supp (\langle c_2 \rangle)$ are disjoint and 
$\supp (\langle c_i \rangle) \subsetneq \supp(\langle c \rangle)$ for $i=1, 2$. In particular, if $c\in C$ has $\wt(c) < 2d$, then $c$ is a minimal codeword of $C$.
\end{lemma}

The following result shows that all ``nontrivial" binary  Reed-Muller codes of order greater than $1$
have a non-pure resolution. 

\begin{proposition}
	\label{thm: binaryRM}
Assume that $m\ge 4$ and $1< r \le m-2$. Then  any minimal free $\mathbb{N}$-resolution of the binary  Reed-Muller code $\mathcal{RM}_2(r,m)$  is not pure.
\end{proposition}

\begin{proof}
The minimum distance of $\mathcal{RM}_2(r,m)$ is $d:= 2^{m-r}$ and if we let 
\[
	Q(X_1,\ldots, X_m)= X_1X_2\cdots X_{r-2}(X_{r-1}X_{r} +X_{r+1}X_{r+2}),
	\]
then clearly, $Q\in V_2(r,m)$. Moreover, the corresponding codeword $c_Q = \Ev(Q)$ has weight $6\times 2^{m-r-2} = 3d/2$. Indeed, $Q(a_1, \dots, a_m) \ne 0$ for $(a_1, \dots, a_m)\in \FF_2^m$ precisely when $a_1= \dots = a_{r-2}=1$, $(a_{r-1}, a_r, a_{r+1}, a_{r+2})$ is one among  $(0,1,1,1)$, $(1,0,1,1)$, $(0,0,1,1)$, $(1,1,0,1)$, $(1,1,1,0)$, and $(1,1,0,0)$, while $a_{r+3}, \dots , a_m \in\FF_2$ are arbitrary. Hence, by Lemma \ref{lem:AB2}, $c_Q$ is a minimal codeword, but it is clearly not of minimum weight. Thus, the desired result follows from Theorem \ref{thm: pure}.
\end{proof}

\begin{remark}
As Alexander Barg has pointed out to one of us, the last assertion in Lemma \ref{lem:AB2} can be extended to the $q$-ary case to show that codewords of weight less than $dq/(q-1)$ are minimal in $C$, where $C$ is a $q$-ary linear code with minimum distance $d$. However, for $q>2$, this is often a restrictive hypothesis, and in the next subsections, we will deal with $q$-ary Reed-Muller codes using a different strategy. 
\end{remark}

\subsection{The Case of $t=0$}\label{subsec:tis0}
Let $t,s$ be as in \eqref{eq:RMnkd} so that $r = t(q-1) + s$ and $0\le s < q-1$. We will consider the case  of 
Reed-Muller codes of order $r >1$ for which $t=0$ (so that $r=s$). Note that such codes are necessarily 
non-binary, and in fact, $q\ge 4$. We shall also exclude the case when $m=1$, since $\mathcal{RM}_q(r,1)$ is a Reed-Solomon (and hence MDS) code for $1\le r \le (q-1)$. 

\begin{proposition}
	\label{thm:RMtiszero}
Assume that  $m\ge 2$ and $1< r < q-1$. Then  any minimal free $\mathbb{N}$-resolution of the  Reed-Muller code $\mathcal{RM}_q(r,m)$  is not pure.
\end{proposition}

\begin{proof}
Choose distinct elements $\omega_1,\ldots, \omega_{r-1}\in\Fq$ and an arbitrary $\omega\in\Fq$. Define
$$
Q(X_1,\ldots, X_m)= (X_2-\omega) \prod\limits_{i=1}^{r-1}(X_1-\omega_i). 
$$
Clearly, $Q\in V_q(r,m)$ and the  corresponding codeword $c_Q = \Ev(Q)$ has weight $(q-r+1)(q-1)q^{m-2}$. On the other hand, by \eqref{eq:RMnkd}, the minimum distance of $\mathcal{RM}_q(r,m)$ is $(q-r)q^{m-1}$.   Observe that 
$$
(q-r+1)(q-1)q^{m-2} - (q-r)q^{m-1} = (r-1)q^{m-2} > 0 \quad \text{since } r>1. 
$$
It follows that $c_Q$ is not a minimum weight codeword. If $c_Q$ is a minimal codeword, then Theorem \ref{thm: pure} implies the desired result. Now suppose $c_Q$ is not a minimal codeword of $\mathcal{RM}_q(r,m)$. Then we can find $F\in V_q(r,m)$ such that $c_F$ is a minimal codeword of $\mathcal{RM}_q(r,m)$ and $\supp(c_F) \subset \supp(c_Q)$. Again, if $c_F$ is not a minimal codeword of $\mathcal{RM}_q(r,m)$, then we are done. Otherwise, by the characterization of  minimum weight  codewords of $\mathcal{RM}_q(r,m)$, we must have
$$
F(X_1,\ldots, X_m) = 
\prod_{j=1}^{r}(L-\omega^\prime_j)
$$
for some 
distinct elements $\omega^\prime_1, \dots , \omega^\prime_r\in \Fq$ and some nonzero linear polynomial  $L$ in $\Fq[X_1, \dots, X_m]$ that we can assume to be homogeneous (by adjusting $\omega^\prime_j$, if necessary). Write $L = a_1X_1+ \cdots +a_mX_m$. Since $\supp(c_F) \subset \supp(c_Q)$, it follows that $L$ vanishes whenever we substitute $X_1=\omega_i$ for some $i\in \{1, \dots , r\}$ or we substitute $X_2=\omega$. In particular, $a_1\omega_1 + a_2X_2 + \cdots +a_mX_m = \omega^\prime_j$ for some $j\in \{1, \dots , r\}$. Comparing the degree in each of the variables $X_2, \dots , X_m$, we obtain $a_2 = \dots = a_m=0$ so that $L=a_1X_1$. But then $L$ does not vanish when we substitute $X_2=\omega$, and we obtain a contradiction. This proves the proposition. 
\end{proof}

\subsection{The case of $0<t<m-1$ and $1<s < q-1$}\label{subsec:tspositive}  
The arguments here will be similar to those in the previous subsection, except that we have to deal with an additional factor of degree $t(q-1)$. Note that $1<s < q-1$ implies that $q\ge 4$. 

\begin{proposition}
	\label{thm:RMtpositive}
Assume that  $1< r < m(q-1)$ and moreover, $r = t(q-1) + s$ with $0<t<m-1$ and $1<s < q-1$. Then  any minimal free $\mathbb{N}$-resolution of the   Reed-Muller code $\mathcal{RM}_q(r,m)$  is not pure.
\end{proposition}

\begin{proof}
Choose distinct elements $\omega_1,\ldots, \omega_{s-1}\in\Fq$ and an arbitrary $\omega\in\Fq$. Define
$$
Q(X_1,\ldots, X_m)= \left(\prod\limits_{i=1}^{t}(X_i^{q-1}-1)\right) \left( \prod\limits_{j=1}^{s-1}(X_{t+1}-\omega_j) \right)(X_{t+2}-\omega)
$$
Clearly, $Q\in V_q(r,m)$ and the  corresponding codeword $c_Q = \Ev(Q)$ has weight $(q-s+1)(q-1)q^{m-t-2}$. On the other hand, by \eqref{eq:RMnkd}, the minimum distance of $\mathcal{RM}_q(r,m)$ is $(q-s)q^{m-t-1}$.   Observe that 
$$
(q-s+1)(q-1)q^{m-t-2} - (q-s)q^{m-t-1} = (s-1)q^{m-t-2} > 0 \quad \text{since } s>1. 
$$
Thus, as in the proof of Proposition~\ref{thm:RMtiszero}, it suffices to show that if there exists $F$ in $V_q(r,m)$ such that 
$c_F$ is a minimum weight codeword with $\supp(c_F)\subset \supp (c_Q)$, then we arrive at a contradiction. Again, any such $F$ has to be of the form
$$
F(X_1, \dots , X_m) =  \left(\prod\limits_{i=1}^{t}(L_i^{q-1}-1)\right) \left( \prod\limits_{j=1}^{s}(L_{t+1}-\omega'_j)\right)
$$
for some 
distinct $\omega^\prime_1, \dots , \omega^\prime_s\in \Fq$, and linearly independent linear polynomials $L_1, \dots , L_{t+1}\in \Fq[X_1, \dots , X_m]$ with $L_{t+1}$ homogeneous. Note that $\supp(c_Q)$ is contained in the linear space $A = \{(a_1, \dots , a_m)\in \Fq^m : a_i = 0 \text{ for } i=1,\dots , t\}$,  which can be identified with $\AA^{m-t}$, 
while $\supp (c_F)$ is contained in the affine space $A':=\{\mathbf{a}\in \Fq^m: L_i(\mathbf{a}) =0 \text{ for } i=1,\dots , t\}$ of dimension $m-t$. Further, since $\supp(c_F)\subset \supp (c_Q)$, we obtain $\supp(c_F)\subseteq A \cap A'$. Now if $A\ne A'$, then $\dim (A\cap A') \le m-t-1$, and so $(q-s)q^{m-t-1} \le q^{m-t-1}$, which is impossible because $s < q-1$. This shows that $A=A'$. Consequently, 
$$
F(0, \dots , 0, X_{t+1}, \dots , X_m) = \prod\limits_{j=1}^{s}\left(L_{t+1}(0, \dots , 0, X_{t+1}, \dots , X_m)-\omega'_j\right)
$$
gives 
a minimum weight codeword in $\mathcal{RM}_q(s,m-t)$ whose support contains the support of the
codeword of $\mathcal{RM}_q(s,m-t)$ associated to $Q(0, \dots , 0, X_{t+1}, \dots , X_m)$. But then this leads to a 
contradiction exactly as in the proof of Proposition~\ref{thm:RMtiszero}.
\end{proof}

\subsection{The case of $s=0$}\label{subsec:siszero} 
Since the binary case and the case $t=0$ have already been dealt with in subsections \ref{subsec:binary} and \ref{subsec:tis0}, we shall assume that $q\ge 3$ and $1\le t \le m-1$. Then $s=0$ implies that $r = t(q-1) > 1$.

\begin{proposition}
	\label{thm:RMsiszero}
Assume that  $q \ge 3$ and $r = t(q-1)$ with $1 \le t \le m-1$. Then  any minimal free $\mathbb{N}$-resolution of the   Reed-Muller code $\mathcal{RM}_q(r,m)$  is not pure.
\end{proposition}

\begin{proof}
Write $\Fq = \{\omega_1,\ldots, \omega_{q}\}$ and pick any 
$\omega\in\Fq$. Consider 
$$
Q(X_1,\ldots, X_m)= \left(\prod\limits_{i=1}^{t-1}(X_i^{q-1}-1)\right) \left( \prod\limits_{j=3}^{q}(X_{t+1}-\omega_j) \right)(X_{t+2}-\omega)
$$
Then $\deg Q = (t-1)(q-1) + (q-2)+1 = t(q-1)=r$ and so $Q\in V_q(r,m)$. Also, we can write
$\supp(c_Q) = A_1\cup A_2$, where for $i=1,2$, 
$$
A_i:=\{\mathbf{a}=(a_1, \dots , a_m)\in \Fq^m: a_1=\dots = a_t=0 , \ a_{t+1} = \omega_i, \text{ and } a_{t+2} \ne \omega\}.
$$  
Clearly, $A_1, A_2$ are disjoint and so 
$\wt(c_Q) = 2(q-1)q^{m-t-1}$. 
The minimum distance of $\mathcal{RM}_q(r,m)$ in this case is $q^{m-t}$, and $2(q-1)q^{m-t-1} > q^{m-t}$, since $q\ge 3$. Thus,  $c_Q$ is not a minimum weight codeword.  As in the proof of Proposition~\ref{thm:RMtiszero}, it suffices to show that the existence of $F\in V_q(r,m)$ such that 
$c_F$ is a minimum weight codeword with $\supp(c_F)\subset \supp (c_Q)$ leads to a contradiction. By the characterization of minimum weight codewords of $\mathcal{RM}_q(r,m)$, any such $F$ has to be of the form
$F(X_1, \dots , X_m) = \prod_{i=1}^{t}(L_i^{q-1}-1)$
for some linearly independent linear polynomials $L_1, \dots , L_t$ in $\Fq[X_1, \dots , X_m]$. Hence, $\supp(c_F)$ is the affine space $A':=\{\mathbf{a}\in \Fq^m: L_i(\mathbf{a}) =0 \text{ for } i=1,\dots , t\}$. Since $\supp(c_F)\subset \supp (c_Q)$, we can argue as in the proof of  Proposition~\ref{thm:RMtpositive} to deduce that $A'$ is in fact, the linear space $\{\mathbf{a}\in \Fq^m: a_1 = \dots = a_t =0 \}$. We now claim that $\supp(c_F)$ is either disjoint from $A_1$ or from $A_2$. Indeed, if this is not the case then there are $P_i\in \supp(c_F)\cap A_i$ for $i=1,2$. But then 
$P_{\lambda}:= P_1 + \lambda(P_2-P_1) \in \supp(c_F)$ for any $\lambda \in \Fq$, since $\supp(c_F) = A'$ is linear. Also since $q\ne 3$, we can pick $\lambda \in \Fq$ such that $\lambda\neq 0$ and $\lambda \neq 1$. Now 
$\supp(c_F)\subset \supp (c_Q) = A_1\cup A_2$ leads to a contradiction since the $t^{\rm th}$ coordinate of $P_{\lambda}$ is neither $\omega_1$ nor $\omega_2$. This proves the claim. It follows that $A'= \supp(c_F)\subset A_i$ for some $i\in \{1, 2\}$. But then $q^{m-t} \le (q-1)q^{m-t-1}$, which is a contradiction. This proves the proposition. 
\end{proof}

\subsection{The case of $t = m-1$ and $1<s < q-2$}\label{subsec:tismminusone}  
We will now consider the last case of nontrivial Reed-Muller codes $\mathcal{RM}_q(r,m)$ of order $r=t(q-1) + s$, where $r>1$ and $s \ne 1$, namely, when $t = m-1$ and $s > 1$. Note that if we allow $s = q-2$, then $\mathcal{RM}_q(r,m)$ becomes a MDS code, and so we shall assume that $1<  s< q-2$. In particular, this implies that $q\ge 5$. 
\begin{proposition}
	\label{thm:RMstsmminusone}
Assume that   $r = (m-1)(q-1)+s$ with $1 < s < q-2$. Then  any minimal free $\mathbb{N}$-resolution of the  Reed-Muller code $\mathcal{RM}_q(r,m)$  is not pure.
\end{proposition}

\begin{proof}
As in the proof of Proposition \ref{thm:RMsiszero}, 
write $\Fq = \{\omega_1,\ldots, \omega_{q}\}$ and pick any 
$\omega\in\Fq$. Also let $\nu_1, \dots , \nu_{s+1}$ be any distinct elements of $\Fq$. Consider 
$$
Q(X_1,\ldots, X_m)= \left(\prod\limits_{i=1}^{m-2}(X_i^{q-1}-1)\right) \left( \prod\limits_{j=3}^{q}(X_{m-1}-\omega_j) \right)  \left( \prod\limits_{j=1}^{s+1} (X_{m}-\nu_j) \right). 
$$
Then $\deg Q = (m-2)(q-1) + (q-2)+(s+1) = (m-1)(q-1)+s=r$ and so $Q\in V_q(r,m)$. Also, $\wt(c_Q) = 2(q-s-1)$ and 
$\supp(c_Q) \subset A_1\cup A_2$, where  $A_i$ denotes the affine line 
$\{\mathbf{a}=(a_1, \dots , a_m)\in \Fq^m: a_1=\dots = a_{m-2}=0 , \ a_{m-1} = \omega_i\}$
for $i=1,2$. The minimum distance of $\mathcal{RM}_q(r,m)$ in this case is $q-s$ and it is less than $2(q-s-1)$, since $s < q-2$. 
As in the proof of Proposition~\ref{thm:RMtiszero}, it suffices to show that the existence of $F\in V_q(r,m)$ such that 
$c_F$ is a minimum weight codeword with $\supp(c_F)\subset \supp (c_Q)$ leads to a contradiction. By the characterization of minimum weight codewords of $\mathcal{RM}_q(r,m)$, any such $F$ has to be of the form
$$
F(X_1, \dots , X_m) = \prod_{i=1}^{m-1}(L_i^{q-1}-1) \prod_{j=1}^s (L_m - \omega'_j)
$$
for some linearly independent linear polynomials $L_1, \dots , L_m$ in $\Fq[X_1, \dots , X_m]$ and distinct $\omega^\prime_1, \dots , \omega^\prime_s\in \Fq$. Also, arguing as in the proof of  Theorem~\ref{thm:RMsiszero}, we see that
$\supp(c_F)$ is contained in the affine line $A':=\{\mathbf{a}\in \Fq^m: L_i(\mathbf{a}) =0 \text{ for } i=1,\dots , m-1\}$. Now if any two points of  $\supp(c_F)$ belong to different affine lines $A_1$ and $A_2$, then $A_i\cap A'$ is nonempty for $i=1,2$ and dimension considerations imply that $A_1=A_2=A'$, which is a contradiction. Hence, the $(q-s)$ points of $\supp(c_F)$ are contained in $\supp (c_Q)\cap A_i$ for a unique $i\in \{1,2\}$. But then $q-s \le q-s-1$, which is a contradiction. This proves the proposition. 
\end{proof}

An easy consequence of the above result is that unlike linear resolutions (which correspond to MDS codes), purity of a resolution is not preserved when passing to the dual. 

\begin{corollary}
	There exist linear codes $C$ with a pure resolution such that $C^\perp$ does not have  a pure resolution.
\end{corollary}

\begin{proof}
By Theorem \ref{thm: firstReedMuller},  the first order Reed-Muller code $\mathcal{RM}_q(1, m)$ has a pure resolution. But the dual of $\mathcal{RM}_q(1, m)$ is $\mathcal{RM}_q((m-1)(q-1) + (q-3), m)$ and it does not have a pure resolution, thanks to Proposition \ref{thm:RMstsmminusone}. 
\end{proof}

We can consolidate the results 
in subsections \ref{subsec:binary}--\ref{subsec:tismminusone} to obtain the following. 

\begin{theorem}
	\label{thm:RMfinal}
Assume that  $m \ge 2$ and $1< r < m(q-1)-1$. Write $r = t(q-1)+s$, where $0 \le t \le m-1$ and $0\le s < q-1$. Suppose $s\ne 1$.  Then  any minimal free $\mathbb{N}$-resolution of the  Reed-Muller code $\mathcal{RM}_q(r,m)$  is not pure.
\end{theorem}

\begin{proof}
Follows from Propositions \ref{thm: binaryRM}, \ref{thm:RMtiszero}, \ref{thm:RMtpositive}, \ref{thm:RMsiszero}, and 
\ref{thm:RMstsmminusone}.
\end{proof}

\section{On the Purity and Resolutions of some two-weight Codes} 

 This section is devoted to two-weight codes. A linear code $C$ is said to be a \emph{ two-weight code} if there are two distinct positive integers $w_1$ and $w_2$ such that every nonzero codeword of $C$ has weight either $w_1$ or $w_2$.  We will usually take $w_1 < w_2$ so that $w_1 = d_1(C)$. We have seen in Corollary~\ref{cor: constwt} that the resolution of constant weight codes are pure and their Betti numbers 
are explicitly known. 
The first order Reed-Muller codes  are  examples of  two-weight codes, and Theorem \ref{thm: firstReedMuller} shows that  their resolutions are pure and the Betti numbers can be explicitly determined. Thus, it is natural to ask if every two-weight code has pure resolution. In this section we will choose several examples of two-weight codes given by Calderbank and Kantor \cite{CK} and see that some of them have pure resolution and others do not. In \cite{CK}, these codes are referred to by a nomenclature such as RT1, TF1, TF$1^d$, etc., and this is indicated in parenthesis at the beginning of each of the examples considered here. We also compute the Betti numbers of some of the two-weight codes irrespective of whether or not their resolution is pure. The examples of two-weight codes given in \cite{CK} are defined geometrically. 
So before considering them 
here, we recall a geometric language for codes and  translate our characterization of purity (Theorem $\ref{thm: pure}$) in 
 this language.

As before, fix positive integers $n,k$ with $k\le n$ and a prime power $q$. 
We denote by $\mathbb{P}^{k-1}$ the $(k-1)$-dimensional projective space over the finite field $\Fq$. 
A (nondegenerate) $[n, k]_q$ \emph{projective system} is a multiset 
of $n$ points in $\mathbb{P}^{k-1}$ that do not lie on a hyperplane of $\mathbb{P}^{k-1}$. 
Let $\P$ be a $[n, k]_q$ projective system. For $r=1, \dots , k$, the $r^{\it th}$ \emph{generalized Hamming weight}, or the  $r^{\it th}$ \emph{higher weight} of  $\P$ is defined by
$$
d_r(\P) =n-\max\{|\P\cap\Pi_r|:\;\Pi_r\ \text{ linear subspace of  }\mathbb{P}^{k-1} \text{ with } \codim\Pi_r=r\}. 
$$
Here the ``cardinality" $|\P\cap\Pi_r|$ is understood as the sum of multiplicities of points of $\P$ that are in $\Pi_r$. Note that the only 
linear subspace 
of codimension $k$ in $\mathbb{P}^{k-1}$ is the empty set, whereas  
those of codimension $k-1$ consist of a single point. 
Thus,    
\begin{equation}
\label{eq:dkP}
d_k(\P) = n \quad \text{and} \quad d_{k-1}(\P) = n-1. 
\end{equation} 
We can naturally associate a nondegenerate $[n,k]_q$-linear code to $\P$ as follows. Choose representatives 
$P_1, \ldots, P_n$ in $\Fq^k$ corresponding to the $n$ points of $\P$. 
Let $(\mathbb{F}_q^k)^*$ be the dual space of the vector space $\mathbb{F}_q^k$. Consider the evaluation map
$$
\Ev: (\mathbb{F}_q^k)^*\to \Fqn\text{ defined by }\Ev(f)= (f(P_1),\ldots, f(P_n)).
$$
The image of $\Ev$ is a linear subspace $C$ of $\Fq^n$ such that $\dim C=k$ and $C$ is not contained in a coordinate hyperplane of $\Fq^n$. This, then, is the  $[n,k]_q$-linear code associated to $\P$. We refer to Tsfasman,  Vl{\u{a}}du{\c{t}} and  Nogin  \cite{TVN} for more on projective systems and simply remark that the above association gives rise to a one-to-one correspondence between the equivalence classes of $[n, k]_q$ projective systems and nondegenerate $[n, k]_q$-linear codes, which preserves generalized Hamming weights. Also, subcodes of $C$ of dimension $r$ correspond to linear subspaces of $\mathbb{P}^{k-1}$ of codimension $r$. Thus, we define the \emph{support} of a linear subspace $\Pi_r$ of $\mathbb{P}^{k-1}$ with $\codim \Pi_r =r$,  to be the multiset 
$\P\setminus \P\cap\Pi_r$. This corresponds precisely to the support of the corresponding subcode of $C$.
As a consequence, we obtain the following geometric translation of our characterization of purity. 


\begin{theorem}
	\label{thm: geompure}
	Let $\P\subseteq \mathbb{P}^{k-1} $ be an $[n, k]_q$ projective system and let $C$ be the corresponding $[n, k]_q$-code. The $\mathbb{N}$-graded minimal free resolution of $C$ is pure if and only if for every $1\leq r\leq  k-1$ and every linear subspace $\Pi_r\subset\mathbb{P}^{k-1}$ of codimension $r$, there exists a  linear subspace $H(\Pi_r)\subset\mathbb{P}^{k-1}$ of codimension $r$ with $\Pi_r\cap\P\subseteq H(\Pi_r)\cap \P$ and 
$|H(\Pi_r)\cap \P| = n-d_r(\P)$.
\end{theorem}
\begin{proof}
Follows from Theorem $\ref{thm: pure}$. 
\end{proof}	

\begin{corollary}
	\label{cor: 2pure}
$\P\subseteq \mathbb{P}^{k-1} $ be an $[n, k]_q$ projective system and let $C$ be the corresponding linear code. 
Then the $\mathbb{N}$-graded resolution of $C$ is always pure at the $(k-1)^{\it th}$ and $k^{\it th}$ step.
\end{corollary}
\begin{proof}
From \eqref{eq:dkP}, we see that $C$ is $(k-1)$-MDS. Thus the desired result  follows from 
Corollary $\ref{Cor: MDS}$ and Theorem \ref{thm: geompure}.
\end{proof}

The following definition from \cite{CK} is a geometric counterpart of two-weight codes. 

\begin{definition}
Let $h_1, h_2$ be distinct nonnegative integers.
An $[n, k]_q$ projective system $\P$  is said to be a \emph{projective $(n, k, h_1, h_2)_q$ system} if 
 every hyperplane of $\mathbb{P}^{k-1}$ intersects $\P$ either at $h_1$ points or at $h_2$ points (counting multiplicities). 
\end{definition}

Note that if $\P$ is a projective $(n, k, h_1, h_2)_q$ system, then every nonzero codeword of the corresponding $[n,k]_q$ code $C$  is  of Hamming weight $w_1$ or $w_2$, where $w_i = n-h_i$ for $i=1,2$. Also note that for $i=1,2$, if $A_{w_i}$ denotes the number of codewords of $C$ of weight $w_i$, then 
\begin{equation} \label{AwN}
A_{w_i} = (q-1)\nu_i, 
\end{equation} 
where $\nu_i$ denotes the number of hyperplanes $\Pi$ of $\mathbb{P}^{k-1}$ such that $|\Pi \cap \P| = h_i$.
The factor $(q-1)$ is due to the fact that the codewords $\Ev (f)$ and $\Ev(\lambda f)$ of $C$ correspond to the same hyperplane in $\mathbb{P}^{k-1}$ for any $\lambda \in \Fq$ with $\lambda \ne 0$. 

We are now ready to discuss several examples from \cite{CK} of two-weight codes, and investigate their purity and minimal free resolutions. We use the following~notation. 
$$
p^{ }_j = p^{ }_j(q):= | \mathbb{P}^{j}(\Fq)| = \begin{cases} q^j + q^{j-1} + \dots + q + 1 & \text{if } j\ge 0, \\
0 & \text{if } j < 0. \end{cases}
$$

\begin{example}[RT1]
	\label{examp: RT1}
Take the base field as $\FF_{q^2}$ and let $\mathbb{P}=\mathbb{P}^{k-1}(\mathbb{F}_{q^2})$. Consider 
$\P=\mathbb{P}^{k-1}(\Fq)$ as a projective system in $\PP$. If $\Pi$ is a hyperplane in $\PP$, then it is given by an equation of the form $\sum_{i=1}^k z_iX_i =0$, where $z_1, \dots , z_k \in \FF_{q^2}$, not all zero. Fix a $\FF_{q}$-basis $\{1, \theta\}$ of $\FF_{q^2}$ and write $z_i = a_i+ \theta b_i$, where $a_i, b_i\in \Fq$ for $i=1, \dots, k$. Then 
$\P\cap \Pi$ consists of points $(c_1: \dots : c_k) \in \PP^{k-1}(\Fq)$ satisfying $\sum a_ic_i =0$ and $\sum b_ic_i =0$. 
Now if there is $\lambda \in \Fq$ such that $a_i = \lambda b_i$ for all $i=1, \dots , k$, or such that $b_i = \lambda a_i$ for all $i=1, \dots , k$, 
then $\P\cap \Pi$ corresponds to a $\Fq$-rational hyperplane in $\PP^{k-1}(\Fq)$. Otherwise, it corresponds to a  linear subspace of codimension 2 in $\PP^{k-1}(\Fq)$. Thus, $|\P\cap \Pi| = p^{ }_{k-2}(q)$ or $p^{ }_{k-3}(q)$. It follows that the linear code corresponding to $\P$, say $C$,  is a two-weight code of length $p^{ }_{k-1}(q)$ and dimension $k$ over $\FF_{q^2}$. Also, it is clear that as $\Pi_r$ varies over $\FF_{q^2}$-linear subspaces of codimension $r$ in $\PP$, the maximum possible value of 
$|\P\cap \Pi_r|$ is attained when $\Pi_r$ is $\Fq$-rational, and in that case $|\P\cap \Pi_r| = p^{ }_{k-1-r}(q)$ for 
$r=1, \dots , k$. It follows that the higher weights of $\P$ are given by $d_r=p^{ }_{k-1}(q) -  p^{ }_{k-1-r}(q)$ for 
$r=1, \dots , k$. 

To determine the purity of the 
minimal free resolution of $C$, fix a  $\FF_{q^2}$-linear subspace $\Pi$ of codimension $r$ in $\PP$. Let $t:=  \dim_{\Fq}(\Pi\cap\P)$. If $\Pi$ is not $\Fq$-rational, then $t < k-1-r$. Let $\{f_1,\ldots, f_{t+1}\}$ be a $\Fq$-basis of $\Pi\cap \P$. Extend this to a linearly independent set $\{f_1,\ldots, f_{t+1},\ldots, f_{k-r}\}\subset \P$. Note that the set $\{f_1,\ldots, f_{k-r}\}$ is linearly independent over $\mathbb{F}_{q^2}$. (This can be seen, as before, by expressing the coefficients in a linear dependence relation in terms of $1, \theta$.) Now if $H = H(\Pi_r)$ is the   linear subspace  of $\PP$ spanned by $\{f_1,\ldots, f_{k-r}\}$, then $\Pi\cap \P\subset H\cap\P$ and 
$|H\cap \P| = n - d_r(\P)$. Thus, Theorem \ref{thm: geompure} shows that the $\mathbb{N}$-graded  minimal free resolution of $C$
is pure. 
Moreover, it is of the form
$$
0\to R(-d_k)^{\beta_{k}}\to\cdots \to R(-d_2)^{\beta_{2}}\to R(-d_1)^{\beta_{1}}\to R
$$
where $d_r=p^{ }_{k-1}(q) -  p^{ }_{k-1-r}(q)$ and $\beta_r$'s are given by Herzog-K{\"{u}}hl equation for $r=1, \dots , k$. 
In fact, this is precisely the resolution for constant weight codes given in  Corollary~\ref{cor: constwt}. It may be noted that even though constant weight codes have been characterized by Johnsen and Verdure \cite[Thm. 2 and Prop. 4]{JV2} as those having a resolution as in Corollary~\ref{cor: constwt}, the code $C$ is not a constant weight code because it is a code over $\FF_{q^2}$, whereas the characterization is for $q$-ary codes. 
\end{example}

\begin{remark}
One can similarly consider $\P=\mathbb{P}^{k-1}(\Fq)\subseteq\mathbb{P}^{k-1}(\mathbb{F}_{q^m})$ for any $m\ge 2$, and show that the resolution of the linear code corresponding to this projective system is pure and of the form similar to that in Example \ref{examp: RT1} even though this code is not a two-weight code when $m>2$. 
\end{remark}

\begin{example}[TF1]
	\label{exam: TF1}
Assume that $q$ is even and consider the projective plane $ \mathbb{P}^2$  over $\Fq$.	Let $\P\subseteq \mathbb{P}^2$ be a \emph{hyperoval}, i.e.,  a set of $q+2$ distinct points, no three collinear, with the property that if $L$ is a line in $ \mathbb{P}^2$, then $|L\cap \P|=0\text{ or }2$. In this case, the corresponding code is an MDS $[q+2, 3]_q$-code and the resolution of this code is given by Corollary \ref{Cor: MDS}. 
\end{example}

\begin{example}[$\text{TF}1^d$]
	\label{exam: TF1d}
Suppose $q$ is even and  $\P $ is the hyperoval in the projective plane 
$\mathbb{P}^2$ over $\Fq$ as in  Example \ref{exam: TF1}.  
Let $\widehat{\mathbb{P}^2}$ be the dual projective plane.  Consider 
$$
\widehat{\P}=\{L: L \text{ is a line in $\mathbb{P}^2$ with  }|L\cap \P|=2 \}. 
$$
Note that $\widehat{\P}\subseteq \widehat{\mathbb{P}^2} $ and the points of the projective plane $\mathbb{P}^2$ are lines in $\widehat{\mathbb{P}^2}$. Note also that any two points of $\P$ correspond to a unique line $L$ in $ \mathbb{P}^2$ such that $L\in \widehat{\P}$. Consequently, 
$|\widehat{\P}|={q+2\choose 2}$. Now consider a line in $ \widehat{\mathbb{P}^2}$, i.e., a point $P$ of  $\mathbb{P}^2$. Counting the intersection of this line with $\widehat{\P}$ corresponds to counting lines $L\subseteq\mathbb{P}^2$ that pass through 
$P$ 
and intersect the hyperoval $\P$ in exactly two points. The cardinality of this intersection depends only on whether  or not the chosen point $P$ lies on $\P$. More precisely, if  $P \in \P$, then any line passing through $P$ will intersect the hyperoval $\P$ in two points,  and there are exactly $(q+1)$ such lines. On the other hand, if $P \not \in \P$, then choosing any point $Q$ on $\P$ will correspond to a unique line $L_Q$  passing through $P$ and $Q$ such that $L_Q$ intersects $\P$ in another point $Q'\ne Q$. Further, since  each $L_Q$ passes through $P$, the points $Q' \in \P$ corresponding to $Q\in \P$ are distinct. Since $|\P|=q+2$, it follows that there are exactly $(q+2)/2$ lines of the form $L_Q$. This shows that $\P$ is a $\left( {q+2\choose 2}, \, 3, \, (q+1), \, \frac{q+2}{2} \right)_q$   projective system, and it corresponds to an $[{q+2\choose 2}, 3]_q$ two-weight code with distinct nonzero weights 
$$
w_1 = {q+2\choose 2} - (q+1) = \frac{q(q+1)}{2} \quad \text{and} \quad w_2 = {q+2\choose 2} - \frac{q+2}{2} = \frac{q(q+2)}{2}. 
$$  
Also, the number of lines in $\widehat{\mathbb{P}^2}$ that intersect $\widehat{P}$ in $(q+1)$  points
is $|\P| = q+2$, whereas the number of lines in $\widehat{\mathbb{P}^2}$ that intersect $\widehat{P}$ in $\frac{q+2}{2}$ points
is $|\mathbb{P}^2 \setminus \P| = q^2-1$. Thus, in view of \eqref{AwN}, we see that 
the weight spectrum of the two-weight code corresponding to $\widehat{\P}$ is given by 
$$
A_{w_1}=(q+2)(q-1) \quad \text{and} \quad 	A_{w_2}=(q^2-1)(q-1).
$$
Furthermore, any hyperplane section of $\widehat{P}$ has to be either of the following two types: (i)  a set consisting of lines passing through a fixed $P\in \P$ and a varying point of $\P\setminus\{P\}$, or (ii) a set consisting of lines of the form $L_Q$ where $Q$ varies over a suitable subset of $\P$ having $(q+2)/2$ elements. Now a set of type (ii) has at least two lines and no two lines in this set can intersect in a point of $\P$. Hence a set of type (ii) can never be contained in any set of type (i). It follows that the purity criterion in Theorem~\ref{thm: geompure} is violated (for $r=1$). Equivalently, every $1$-dimensional subcode of $C$ has minimal support, and since $C$ has two distinct nonzero weights $d_1=w_1 < w_2$, we see that the criterion in Theorem \ref{thm: pure} is violated (for $i=1$). Thus, the resolution of $C$ is not pure. Moreover, in view of \eqref{eq: betaij-betasigma} and \eqref{eq:betaisigma}, we see that the resolution  has two twists at the first step, whereas it is pure at the second and third step, thanks to Corollary~\ref{cor: 2pure}. Hence, the resolution of $C$ is of the form 
$$
R(-d_3)^{\beta_{3, d_3}}\to R(-d_2)^{\beta_{2, d_2}}\to R(-w_2)^{\beta_{1, w_2}}\oplus R(-w_1)^{\beta_{1, w_1}} 
$$
where $w_1, w_2$ are as before and 
$$
d_2 = {q+2\choose 2} - 3 +2 = \frac{q(q+3)}{2} \quad \text{and} \quad  d_3= {q+2\choose 2} - 3 +3 = \frac{(q+1)(q+2)}{2}.
$$
Moreover, from Corollary~\ref{Cor:beta1}, we see that 
$$
\beta_{1, w_1} = (q+2)  \quad \text{and} \quad \beta_{1, w_2}
=(q^2-1).
$$
To determine the remaining Betti numbers, let us write $X_1=\beta_{1, w_1}$, $X_2=\beta_{1, w_2}$, $Y=\beta_{2, d_2}$, and $Z=\beta_{3, d_3}$. Then  the Boij-S\" oderberg equations~$\eqref{eq: boijsoderberg}$~give the following system of linear equations
	\begin{eqnarray*}
	1 -( X_1 + X_2) + Y - Z=0 \\
	- w_1X_1 - w_2 X_2 + d_2Y - d_3Z=0 \\
	- w_1^2X_1 - w_2^2 X_2 + d_2^2Y - d_3^2Z=0
	\end{eqnarray*}
	Putting the values of $w_1,\; w_2, \; d_2, \;d_3,\; X_1$ and $X_2$, we obtain 
$Y=\frac{q(q+1)(q+2)}{2}$ and $Z=\frac{q^2(q+1)}{2}$. This determines the resolution of the code $C$ corresponding to $\widehat{\P}$. 
\end{example}

\begin{example}[TF2]
	\label{examp: TF2}
Assume that $q$ is even with $q>2$. 
Suppose $h$ is an integer such that $1<h<q$ and $h$ divides $q$. 
Following Denniston \cite{Den}, a \emph{maximal arc}  in the projective plane $\PP^2$ may be defined as a set of points meeting every line in $h$ points or none at all.  
Let $\P\subseteq\mathbb{P}^2$ be a maximal arc consisting of $n= 1+(q+1)(h-1)$ points.
It has been shown by Denniston \cite{Den} that such maximal arcs exist. Since $|L\cap\P|= 0\text{ or }h$, for any line $L$ in $\PP^2$, we see that the $[n,3]_q$-code $C$ corresponding to $\P$ is a two-weight code (cf. \cite{CK}) whose nonzero weights are  $q(h-1)$ and $n$. 
Since the second weight of $C$ is the length of $C$, a minimal $1$-dimensional subcode of $C$ must be of minimum weight. Hence, by Theorem~\ref{thm: pure}, the minimal free resolution of the code $C$ is pure. 
Thus, in view of Corollary~\ref{cor: 2pure}, we see that the resolution of $C$ is of the form:
	$$
	R(-d_3)^{\beta_{3, d_3}}\to R(-d_2)^{\beta_{2, d_2}}\to R(-d_1)^{\beta_{1, d_1}}
	$$
	where $d_1=q(h-1)$, $d_2=(q+1)(h-1)$ and $d_3= 1+(q+1)(h-1)$. Using the Herzog-K{\"{u}}hl formula, one can compute the Betti numbers, and they are
	$$
	\beta_{1, d_1}=(q+1)^2-\frac{q}{h},\quad \beta_{2, d_2}=qn, \quad \text{and}\quad \beta_{3, d_3}=(h-1)^2(q+1)\frac{q}{h}.
	$$
\end{example}	
	
\begin{example}[$\text{TF}2^{d}$]
	\label{examp:TF2d}
Let $q, h, n$ and $\P$ be as in Example~\ref{examp: TF2}. Consider the dual  projective plane $\widehat{\mathbb{P}^2}$ of $\mathbb{P}^2$, and let $\widehat{\P}=\{L\in\widehat{\mathbb{P}^2}: |L\cap \P|= h\}$.  
Now there are exactly $(q+1)$ lines passing through a point of $\PP^2$, and in case this point is in $\P$, then such a line intersects $\P$ in exactly $h$ points. Since $|\P|= n$, 
it follows that 
$$
 \hat{n}:= |\widehat{\P}|= \frac{(q+1)n}{h} =  \frac{(q+1)\left( 1+(q+1)(h-1)\right)}{h}.
$$
Next we want to understand the intersection of $\widehat{\P}$ with a hyperplane of $\widehat{\mathbb{P}^2}$. Note that a hyperplane, say $H$, of $\widehat{\mathbb{P}^2}$ corresponds to a point, say $P$, of  $\mathbb{P}^2$, and 
$$
H\cap \widehat{\P}=\{L\subset \mathbb{P}^2:\; L\text{ is a line passing through }P\text{ and }|L\cap\P|=h\}. 
$$
Therefore 	$|H\cap \widehat{\P}|$ is $(q+1)$ or $n/h$, according as $P\in \P$ or $P\not\in \P$.  Thus $\widehat{\P}$ is an 
$(\hat{n}, \; 3, \; (q+1), \; \frac{n}{h})_q$ projective system  
and 
the corresponding $[\hat{n}, \; 3]_q$-code is a  two-weight code  with distinct nonzero weights given~by 
$$
w_1= \hat{n} - (q+1)  = \frac{q(q+1)(h-1)}{h} \quad \text{and} \quad w_2=\hat{n} - \frac nh = \frac{qn}{h}.
$$
Using similar arguments as in Example~\ref{exam: TF1d}, we see  that the weight spectrum of this code is given by 
$$
	A_{w_1}=(q-1)n  \quad \text{and} \quad A_{w_2}=(q-1)(q+1)(q - h +1), 
$$
and also that the resolution of this code is of the form
$$
R(-d_3)^{\beta_{3, d_3}}\to R(-d_2)^{\beta_{2, d_2}}\to R(-w_2)^{\beta_{1, w_2}}\oplus R(-w_1)^{\beta_{1, w_1}} 
$$	
where $\beta_{1, w_1}=n=  1+(q+1)(h-1)$ and $\beta_{1, w_2}=(q+1)(q - h +1)$, and in view of Corollary~\ref{cor: 2pure}, $d_2 =  \hat{n} - 1$ and $d_3 =  \hat{n}$. 
As in Example~\ref{exam: TF1d}, 
using the Boij-S\" oderberg equations $\eqref{eq: boijsoderberg}$ and putting all known values, we obtain
$$
\beta_{2, d_2}=\frac{q(q+1)(qh +h-q)}{h} \quad \text{and} \quad \beta_{3, d_3}= \frac{q^2(q+1)(h-1)}{2}.
$$

\end{example}

We remark that when $q>2$, Examples \ref{exam: TF1} and \ref{exam: TF1d} are special cases of 
Examples~\ref{examp: TF2} and \ref{examp:TF2d}, respectively, with $h=2$. 

\begin{example}[TF3]
	\label{examp: TF3}
Assume that $q > 2$. 
In the finite projective $3$-space $\PP^3$ over $\Fq$, an \emph{ovoid} may be defined as a set of $q^2+1$ points,  no three of which are collinear (see, e.g., Dembowski \cite[p. 48]{Demb}).
Suppose $\P$ is an ovoid in $\PP^3$. Then for any hyperplane $H$ of $\PP^3$, the intersection $\P\cap H$ is an ovoid in $H\simeq \PP^2$, and hence using  \cite[p. 48, \S 49]{Demb}, we see that $|H\cap\P|= 1\text{ or }q+1$. 
Let $C$ be the corresponding linear code. Then $C$ is a two-weight code of length $n=q^2+1$, dimension $k=4$,  and weights $w_1=q(q-1)$ and $w_2=q^2$. The resolution of this code $C$ is pure. To see this, note that if $\Pi$ is a hyperplane in $\mathbb{P}^3$ intersecting $\P$ at only one point, then there is another hyperplane $H$ with $|H\cap \P|= q+1$ and $\Pi\cap\P\subset H\cap \P$. More precisely, let $\Pi\cap \P=\{P\}$ and let $Q\in\P$ be any point other than $P$. Take any hyperplane $H$ passing through $P$ and $Q$. Since $|H\cap \P| > 1$, we must have $|H\cap \P|=q+1$. Further, $\Pi\cap\P\subset H\cap \P$. It follows that  all minimal codewords of $C$ are of minimum weight. Hence, by Corollary~\ref{Cor:beta1}, we see that $\beta_{1,j} = 0$ for all $j \ne w_1$, i.e., 
the resolution of $C$ is ``pure at the first step". Next, observe that the maximum possible cardinality of $L\cap \P$ is $2$ for any line $L$ in $\PP^3$, and there do exist lines $L$ for which $|L\cap \P|=2$. Hence,  
	$
	d_2(C)= n-2 = q^2-1.
	$
	Consequently, $C$ is a $2$-MDS code, and hence by Corollary~\ref{Cor: MDS}, the resolution is linear after the second step.  This proves that the resolution of $C$ is pure and is of the form
	$$
	R(-(q^2+1))^{\beta_{4, q^2+1}}\to R(-q^2)^{\beta_{3, q^2}}\to R(-(q^2-1))^{\beta_{2, q^2-1}}\to R(-q(q-1))^{\beta_{1, q(q-1)}}
	$$
where the Betti numbers can be obtained 
from  Herzog-K{\"{u}}hl formula \eqref{eq: herzogkuhl} as follows. 
	\begin{eqnarray*}
		& \beta_{4, q^2+1}=\frac{q^3(q-1)^2}{2},\;\quad \quad  \quad  &\ \beta_{3, q^2}=(q-1)(q^2-1)(q^2+1), \\
		& \beta_{2, q^2-1}=\frac{q^3(q^2+1)}{2},\quad\text{and } \quad & \beta_{1, q(q-1)}=q(q^2+1). 
	\end{eqnarray*}

\end{example}

\begin{example}[RT3]
	\label{examp: Hermitian}
Assume that $k \ge 3$. 
Consider the quadratic extension $\mathbb{F}_{q^2}$ of $\Fq$ and the projective variety 
$\P^{ }_{k-2} \subset\mathbb{P}^{k-1}(\mathbb{F}_{q^2})$ 
defined by the equation
	$$
	X_1^{q+1}+\cdots+X_k^{q+1}=0.
	$$
Following Bose and Chakravarti~\cite{BC}, we may refer to $\P^{ }_{k-2}$ as the (nondegenerate) Hermitian variety of dimension $k-2$. Let $C^{ }_{k-2}$ be the 
$[n^{ }_k, k]^{ }_{q^2}$-code corresponding to $\P^{ }_{k-2}$, where $n^{ }_{k} := |\P^{ }_{k-2}|$. 
We know from \cite[Theorem 8.1]{BC} that
\begin{equation}\label{eq:nk}
n^{ }_{k} 
= \frac{\left(q^k - (-1)^k \right) \left(q^{k-1} - (-1)^{k-1}\right) } {q^2 -1}.
\end{equation}
To understand the weights of $C^{ }_{k-2}$, first note that since $x\mapsto x^q$ is an involutory automorphism of $\FF_{q^2}$, every hyperplane of $\PP^{k-1}(\FF_{q^2})$ is given by an equation of the form  $c_1^q X_1 + \dots + c_k^q X_k =0$ for some $\mathbf{c} = (c_1: \dots : c_k) \in \PP^{k-1}(\FF_{q^2})$; we denote this hyperplane by $H_{\mathbf{c}}$ and call it a \emph{tangent hyperplane} in case $\mathbf{c}\in \P^{ }_{k-2}$ (see, e.g., Chakravarti  \cite[\S 2]{C}). We remark that $H_{\mathbf{c}}$ and ${\mathbf{c}}$ determine each other. In other words, if ${\mathbf{c}}, {\mathbf{d}} \in \PP^{k-1}(\FF_{q^2})$, then:  $H_{\mathbf{c}} = H_{\mathbf{d}} \Leftrightarrow {\mathbf{c}} = {\mathbf{d}}$. 
Now from \cite[Theorem 3.1]{C} and from Theorem 7.4 as well as Theorem 8.1 (and its corollary) of \cite{BC}, we see that 
\begin{equation}\label{eq:hypsec}
|H_{\mathbf{c}} \cap \P^{ }_{k-2}| = \begin{cases} n_{k-1} & {\rm if } \ H_{\mathbf{c}} \text{ is not a tangent hyperplane,} \\
1+ q^2n_{k-2} & {\rm if } \  H_{\mathbf{c}} \text{ is a tangent hyperplane,} 
\end{cases}
\end{equation}
where $n_{k-1}$ and $n_{k-2}$ are given by expressions similar to that in \eqref{eq:nk} with appropriate substitution. 
Thus, it follows that $C^{ }_{k-2}$ is a two-weight code. We will now discuss the nature of the resolution of this code when $k=3$ and $k=4$.

First, suppose $k=3$. Then $\P_1$ is the Hermitian curve consisting of $q^3+1$ points. If $L$ is a line in $\mathbb{P}^{2}(\mathbb{F}_{q^2})$, then by \eqref{eq:hypsec}, $|L\cap \P_1|$ is either $q+1$ or $1$, and thus the two nonzero weights of $C_1$ are given by 
$w_1 
= q(q^2-1)$ and $w_2 = q^3$.  Moreover, if $L_1$ is a tangent line to $\P_1$ so that $L_1\cap \P_1$ consists of a single point, say $P$, 
then by choosing another point $Q$ of $\P_1$ and a line $L_2$ passing through $P$ and $Q$, we find 
$$
|L_2\cap \P_1|=q+1 \quad \text{and} \quad  	L_1\cap\P_1\subset L_2\cap \P_1.
$$
Consequently, every $1$-minimal subcode of $C_1$ has support weight $w_1 = d_1(C_1)$. Thus, as in Example~\ref{examp: TF3}, we can deduce from Corollary~\ref{Cor:beta1} that the resolution of $C_1$ is ``pure at the first step". This together with Corollary \ref{cor: 2pure} shows that the resolution of $C_1$ is pure and it looks like
$$
 R(-(q^3+1))^{\beta_{3, q^3 +1}}\to R(-q^3)^{\beta_{2, q^3}}\to R(-q(q^2-1))^{\beta_{1, q(q^2-1)}}
$$
where the Betti numbers can be obtained 
from  Herzog-K{\"{u}}hl formula \eqref{eq: herzogkuhl} as follows.
$$
{\beta_{1, q(q^2-1)}}= q^2(q^2-q+1),\; {\beta_{2, q^3}}= (q^3+1)(q^2-1)\text{ and }{\beta_{3, q^3 +1}}= q(q^2-1)(q^2-q+1).
$$
	
Next, suppose $k=4$. Here $\P_2$  is the Hermitian surface 
with $(q^2+1)(q^3+1)$ points. Further, by \eqref{eq:hypsec}, a section $\P_2\cap H_{\mathbf{c}}$ of the Hermitian surface by a tangent hyperplane has $ q^3 +q^2 +1 $ points, while a section $\P_2\cap H_{\mathbf{d}}$ by a non-tangent hyperplane has $q^3+1$ points. Moreover, $\P_2\cap H_{\mathbf{d}} \not\subseteq \P_2\cap H_{\mathbf{c}}$ for any $\mathbf{c} \in \P_2$ and  $\mathbf{d} \in \PP^3(\FF_{q^2}) \setminus \P_2$. Indeed, by \cite[Theorem 3.1]{C}, $\P_2\cap H_{\mathbf{d}}$ is nondegenerate in $H_{\mathbf{d}} \simeq \PP^2$ and so the linear span of points in 
$\P_2\cap H_{\mathbf{d}}$ 
is  $H_{\mathbf{d}}$. But then $\P_2\cap H_{\mathbf{d}} \subseteq \P_2\cap H_{\mathbf{c}}$ would imply that 
$ H_{\mathbf{d}} \subseteq H_{\mathbf{c}}$ and hence $H_{\mathbf{d}} = H_{\mathbf{c}}$, which is a contradiction.  (Alternatively, if $\P_2\cap H_{\mathbf{d}} \subseteq \P_2\cap H_{\mathbf{c}}$, then 
$q^3+q^2+1 =|\P_2\cap H_{\mathbf{d}}| = |\P_2\cap H_{\mathbf{d}} \cap H_{\mathbf{c}}| \le |H_{\mathbf{d}} \cap H_{\mathbf{c}}| =q^2+1$, which is  a contradiction.) At any rate, it follows that $C_2$ is a two-weight code with the nonzero weights $w_1= q^5$ and $w_2=q^5+q^2$, and moreover, every $1$-dimensional subcode of $C_2$ is minimal. Thus, the resolution of $C_2$ has two twists at the first level and by Corollary \ref{Cor:beta1}, 
the corresponding Betti numbers 
are as follows.
$$
\beta_{1, w_1} = |\P_2| =  (q^2+1)(q^3+1) \quad \text{and} \quad  
\beta_{1, w_2} = | \PP^3(\FF_{q^2}) \setminus  \P_2| = 
q^3(q^2+1)(q-1).
$$
To understand the behavior of the resolution at the second step, we consider $2$-dimensional subcodes of $C_2$ and determine which of these are minimal. Equivalently, we consider the sections $\P_2\cap L$ of the Hermitian surface with a line $L$ in $ \PP^3(\FF_{q^2})$. It is shown in \cite[\S 10]{BC} (see also \cite[\S 5.2]{C}) that 
$|\P_2\cap L|$ can only take $3$ possible values, namely, $q^2+1$,  $q+1$, or $1$.  
Accordingly, the line $L$ is referred to as a generator, secant line, or tangent line, respectively. It is clear that if $L$ is a tangent line, then there is a non-tangent line $L'$ such that $\P_2\cap L \subset \P_2\cap L'$. On the other hand, if $L$ is a secant line, then $\P_2\cap L \not\subset \P_2\cap L'$ for any generator $L'$, because there is a unique line passing through any two points of  $\PP^3(\FF_{q^2})$.   It follows that there are two types of $2$-minimal subcodes of $C_2$, one with support weight $d_2=|\P_2| - (q^2+1) = q^3(q^2 + 1)$ and another with support weight  $d_2^\prime= |\P|-(q+1)=q(q^4 +q^2 +q -1)$.  Thus, it follows from \eqref{eq: betaij-betasigma} and \eqref{eq:betaisigma} that the resolution of $C_2$ has two twists at level $2$, and these correspond to the above values of $d_2$ and $d_2'$. Finally, we note that $C_2$ is $3$-MDS and by Corollary \ref{cor: 2pure}, the resolution of $C_2$ is pure at the third and fourth steps. Thus, we can conclude that the minimal free resolution of $C_2$ has the form 
$$
	R(-d_4)^{z}\to R(-d_3)^{y}\to R(-d_2^\prime)^{x_1}\oplus R(-d_2)^{x_2}\to R(-w_2)^{\beta_{1, w_2}}\oplus R(-w_1)^{\beta_{1, w_1}}
	$$
	where $w_1,w_2, d_2,d_2^\prime$ are as before, $d_3= (q^2+1)(q^3+1)-1$, $d_4= (q^2+1)(q^3+1)$, and 
$x_1, x_2, y, z$ denote the undetermined Betti numbers, namely, 
$$
x_1 = \beta_{2, d_2^\prime}, \quad x_2 = \beta_{2, d_2}, \quad y = \beta_{3, d_3}, \quad 
 z = \beta_{4, d_4}.
$$ 
To determine these, we note that the Boij-S\"oderberg equations \eqref{eq: boijsoderberg} give rise to
\begin{eqnarray*}
1 - (\beta_{1, w_1}+ \beta_{1, w_1}) + (\beta_{2, d_2}+ \beta_{2, d_2'}) -\beta_{3, d_3} +\beta_{4, d_4}=0\\
-(w_1\beta_{1, w_1}+ w_2\beta_{1, w_1}) + (d_2\beta_{2, d_2}+ d_2'\beta_{2, d_2'}) -d_3\beta_{3, d_3} +d_4\beta_{4, d_4}=0\\
-(w_1^2\beta_{1, w_1}+ w_2^2\beta_{1, w_1}) + (d_2^2\beta_{2, d_2}+ d_2'^2\beta_{2, d_2'}) -d_3^2\beta_{3, d_3} +d_4^2\beta_{4, d_4}=0\\
-(w_1^3\beta_{1, w_1}+ w_2^3\beta_{1, w_1}) + (d_2^3\beta_{2, d_2}+ d_2'^3\beta_{2, d_2'}) -d_3^3\beta_{3, d_3} +d_4^3\beta_{4, d_4}=0
\end{eqnarray*}
and this is a system of four linear equation in four unknowns. Substituting the values of the known quantities and solving, we obtain 
\begin{eqnarray*}
\beta_{2, d_2}& =& q^2(q^3+1)(q+1),\qquad \qquad \qquad \qquad \ \ \beta_{2, d'_2}=q^6(q^2+1)(q^2-q+1), \\
	\beta_{3, d_3}&=&q^3(q^2+1)(q^3+1)(q^3-q+1), \ \text{ and }\ \beta_{2, d_4}=q^9(q^2-q+1).
\end{eqnarray*}
Thus, the resolution of $C_2$ is completely determined. 

We remark that when $k\ge 5$, the minimum weight of $C_{k-2}$ will be $n_k - n_{k-1}$ or $n_k - 1 - q^2n_{k-2}$ according as $k$ is odd or even. Moreover, a cardinality argument similar to the one in the case of $k=4$ will show that all $1$-dimensional subcodes of $C_{k-2}$ are minimal, and hence the resolution is not pure (at the first step). It would be interesting to completely determine the resolution of $C_{k-2}$, in general. 
\end{example}

\section*{Acknowledgments}
We are grateful to Trygve Johnsen for helpful discussions and encouragement. The authors are also grateful to the 
Indo-Norwegian project MAIT: EECC, supported by the Research Council of Norway and the Dept. of Science and Technology of Govt. of India, which facilitated mutual visits that were helpful to complete this work. 
The second named author would like to acknowledge past funding, during earlier stages of this work, 
 from the Council of Scientific and Industrial Research, India
for  a doctoral fellowship at IIT Bombay, and from 
H. C. \O{}rsted  cofund postdoctoral fellowship at the Technical University of Denmark. 


\end{document}